\documentclass[onecolumn]{IEEEtran}
\IEEEoverridecommandlockouts                              % This command is only
\usepackage{graphicx,cite} % for pdf, bitmapped graphics files
\usepackage{amsmath, amssymb}

\newtheorem{theorem}{Theorem}
\newtheorem{proposition}{Proposition}

\usepackage{standalone, amsmath, amssymb, graphicx, color}
\usepackage{url}
\usepackage{caption}
\usepackage[scriptsize,tight]{subfigure}

\newcommand{\vo}[1]{\boldsymbol{#1}} % Vector object
\newcommand{\mo}[1]{\boldsymbol{#1}} % Matrix object

\newcommand{\real}{\mathbb{R}}

\newcommand{\Exp}[1]{\boldsymbol{\mathsf{E}} \left[#1\right]}

\newcommand{\uniform}[1]{\mathcal{U}_{#1}}

\newcommand{\x}{\vo{x}} % state
\newcommand{\xdot}{\dot{\vo{x}}} % state derivative
\renewcommand{\u}{\vo{u}} % output
\newcommand{\param}{\vo{\Delta}} %parameters
\newcommand{\domain}[1]{{\mathcal{D}_{#1}}} %parameters
\newcommand{\Y}{\vo{Y}} % Multiple observations
\newcommand{\A}{\mo{A}} % A matrix of linear system
\let\B=\undefined
\newcommand{\B}{\mo{B}} % B matrix of linear system

\usepackage{xifthen}% provides \isempty test
\newcommand{\basis}[2]{%
 \phi_{#1}
  \ifthenelse{\isempty{#2}}%
    {}% if #1 is empty
    {({#2})}% if #1 is not empty
}

\newcommand{\xpc}{\x_{pc}} % state
\newcommand{\xpcdot}{\dot{\x}_{pc}} % state
\newcommand{\Apc}{\mo{A}_{pc}} % A matrix
\newcommand{\I}[1]{\vo{I}_{{#1}}} % identity

\renewcommand{\vec}[1]{\boldsymbol{\mathsf{vec}}\left({#1}\right)}

\newcommand{\X}{\mo{X}} % state
 
\newcommand{\W}{\mo{W}} 
\newcommand{\K}{\mo{K}} % Feedback gain.

\newcommand{\set}[1]{\mathcal{#1}}
\newcommand{\inner}[1]{\left\langle #1 \right\rangle}

\newcommand{\figlabel}[1]{\label{fig:#1}}
\newcommand{\eqnlabel}[1]{\label{eqn:#1}}

\newcommand{\eqn}[1]{\eqref{eqn:#1}}

\newcommand{\fig}[1]{fig.(\ref{fig:#1})}
\newcommand{\Fig}[1]{Fig.(\ref{fig:#1})}

\newcommand{\etal}{\textit{et al. }}

\definecolor{darkgreen}{rgb}{0,0.65,0}

\renewcommand{\P}{\mo{P}} % state
\newcommand{\Acl}{\A_\textit{cl}}
\newcommand{\VK}{\mo{V}_{\K}}

\newcommand{\mVK}{\set{V}_{\K}}

\newcommand{\Phin}[1]{\mo{\Phi}_{#1}}

\newcommand{\vPhi}{\mo{\Phi}}
\renewcommand{\AA}{\Exp{\Phin{n}\A\Phin{n}^T}}
\newcommand{\BB}{\Exp{\Phin{n}\B\Phin{m}^T\Phin{m\times(N+1)}^T}}

\title{A Polynomial Chaos Framework for Designing Linear Parameter Varying Control Systems}
\author{Raktim Bhattacharya%
\thanks{Raktim Bhattacharya is with Faculty of Aerospace Engineering, Texas A\&M University, College Station, TX 77843-3141. Email:{\tt\small raktim@tamu.edu}}}

\begin{document}
\maketitle
\thispagestyle{empty}
\pagestyle{empty}

\begin{abstract}
Here we use polynomial chaos framework to design controllers for linear parameter varying (LPV) dynamical systems. We assume the scheduling variable to be random and use polynomial chaos approach to synthesize the controller for the resulting linear stochastic dynamical system. The stability of the LPV system is formulated as an exponential mean-square (EMS) stability   problem. Two algorithms are presented that guarantee EMS stability of the stochastic system and correspond to parameter dependent and independent Lyapunov functions, respectively. LPV controllers from the polynomial chaos based framework is shown to outperform LPV controller from classical design for an example nonlinear system.
\end{abstract}

\section{Introduction}
Linear parameter varying (LPV) systems are of the form 
\begin{equation}
\xdot = \A(\rho)\x + \B(\rho)\u,\eqnlabel{lpv}
\end{equation}
where system matrices depend on unknown parameter $\rho(t)$, which is measurable in real-time \cite{shamma2012overview, leith2000survey}. Many nonlinear systems can be transformed to LPV systems and control systems can be designed using parameter dependent convex optimization problems. Typically, parameter dependent quantities are approximated using a known class of functions such as multilinear basis functions of $\rho$, linear fractional transformations of system matrices, or by gridding the parameter space. Both these approaches result in solution of a finite, but possible large, number of linear matrix inequalities (LMIs). Further, the choice of the basis functions or the resolution of the grid could lead to conservatisms in the design. Clearly, there is a tradeoff between problem size and conservatism in the design \cite{toker1997complexity}. 

Fujisaki \etal \cite{fujisaki2003probabilistic} addressed the computational complexity of such problems by presenting a probabilistic approach to solve these problems, via a sequential randomized algorithm, which significantly reduces the computational complexity. Here the parameter $\rho(t)$ is assumed to be bounded i.e. $\rho(t) \in \domain{\rho} \subset \real^d$ and is treated as a random variable, with a distribution $f_{\rho}(\rho)$ defined over $\domain{\rho}$. The LPV synthesis problem is solved by sampling $\domain{\rho}$ and solving the sampled LMIs using a sequential-gradient method. As with any probabilistic algorithm, there is a tradeoff between sample complexity and confidence in the solution. Often, a large number of samples are required to generate a solution with high confidence.  Also, the LMIs depend only on $\rho(t)$ and not in $\dot{\rho}(t)$ as it is in classical LPV formulation.

This paper is motivated by the work of Fujisaki \etal and is based on the idea of treating $\rho$ as a random variable. Therefore, by substituting $\rho\equiv \param$ in the system equation, we get 
\begin{align}
 \dot{\x}	 &=  \A(\param)\x + \B(\param)\u, \eqnlabel{contiDyn}
\end{align}
where $\param\in\real^d$ is a vector of uncertain parameters, with joint probability density function $f_{\param}(\param)$. Matrices $\A(\param)\in\real^{n\times n}$, $\B(\param) \in\real^{n\times m}$ are system matrices that depend on $\param$. Consequently, the solution $\x:=\x(t,\param)\in\real^n$ also depends on $\param$. Like in \cite{fujisaki2003probabilistic} we ignore temporal variation in the parameter and thus treat $\param$ as random variables. Thus, we now study the system in \eqn{lpv} as a \textit{linear time invariant system} with probabilistic system parameters. The LPV control design objective is then equivalent to designing a state-feedback law of the form $\u = \K(\param)\x$, which stabilizes the system in some suitable sense, where $\K(\param)\in\real^{m\times n}$. Thus, we are looking to obtain a parameter dependent gain $\K(\param)$ that stabilizes the system in \eqn{contiDyn}. The closed-loop system is then
\begin{align}
\dot{\x}	 &=  \left[\A(\param) + \B(\param)\K(\param)\right]\x, \nonumber \\
			 &=  \Acl\x. \eqnlabel{Acl}
\end{align}

There are two distinct differences between the work presented here and that in \cite{fujisaki2003probabilistic}. We do not use a randomized approach to solve the stochastic problem, and thus don't have issues related to confidence in the solution. In our approach, the stochastic problem is solved using polynomial chaos theory, which is a deterministic approach as described later. In addition, stability of the LPV system is formulated as an exponential mean square stability problem for the corresponding stochastic system. In \cite{fujisaki2003probabilistic}, stability of the LPV system is formulated in the probabilistic sense. Computationally, the polynomial chaos framework is superior to sampling based approach in propagating uncertainty \cite{le2010spectral}, and hence we can expect a computational advantage in using this framework to solve the stochastic formulation.

Main contributions of this paper are two LPV controller synthesis algorithms with parameter dependent and independent Lyapunov functions respectively. They are presented as theorem 1 and 2. The paper is organized as follows. We first provide a brief background on polynomial chaos theory and show how it is applied to study linear dynamical systems with random parameters. This is followed by conditions for exponential mean-square stability in the polynomial chaos framework for closed-loop systems with parameter dependent controller. This leads to theorem 1 and 2. The paper ends with an example that highlights the superiority of the polynomial chaos approach over the classical LPV design approach.

\section{Polynomial Chaos Theory}
%Polynomial chaos is a non-sampling based method to determine evolution of uncertainty in dynamical system with probabilistic system parameters \cite{wiener,pcFEM}.  
Polynomial chaos (PC) is a non-sampling based method to determine evolution of uncertainty in dynamical system, when there is probabilistic uncertainty in the system parameters. Polynomial chaos was first introduced by Wiener \cite{wiener}
where Hermite polynomials were used to model stochastic processes
with Gaussian random variables. It can be thought of as an extension of Volterra's theory of nonlinear functionals for stochastic systems \cite{volterra,pcFEM}. According to Cameron and Martin \cite{CameronMartin} such an expansion converges in the $\mathcal{L}_2$ sense for any arbitrary stochastic process with
finite second moment. This applies to most physical systems. Xiu
\etal \cite{xiu:02} generalized the result of Cameron-Martin to various
continuous and discrete distributions using orthogonal polynomials
from the so called Askey-scheme \cite{Askey-Polynomials} and
demonstrated $\mathcal{L}_2$ convergence in the corresponding Hilbert
functional space. The PC framework has been applied to applications including
stochastic fluid dynamics \cite{pcFluids2,pcFluids4,pcFluids5},
stochastic finite elements \cite{pcFEM}, and solid mechanics
\cite{pcSolids1,pcSolids2}, feedback control \cite{hover2006application, kim2012generalized, fisher2009linear, bhattacharya2012linear} and estimation \cite{Dutta2010}. It has been shown that PC based methods are computationally far superior than Monte-Carlo based methods \cite{xiu:02, pcFluids2, pcFluids4, pcFluids5, le2010spectral}. See \cite{eldred2009comparison} for several benchmark problems.

A general second order process $X(\omega)\in
\mathcal{L}_2(\Omega,\mathcal{F},P)$ can be expressed by polynomial
chaos as
\begin{equation}
\eqnlabel{gPC}
X(\omega) = \sum_{i=0}^{\infty} x_i\phi_i({\param}(\omega)),
\end{equation}
where $\omega$ is the random event and $\phi_i({\param}(\omega))$
denotes the polynomial chaos basis of degree $p$ in terms of the random variables
$\param(\omega)$. $(\Omega,\mathcal{F},P)$ is a probability space, where $\Omega$
is the sample space, $\mathcal{F}$ is the $\sigma$-algebra of the
subsets of $\Omega$, and $P$ is the probability measure. According to Cameron and Martin \cite{CameronMartin} such an expansion converges in the $\mathcal{L}_2$ sense for any arbitrary stochastic process with finite second moment. In practice, the infinite series is truncated and $X(\omega)$ is approximated by 
\[
X(\omega) \approx \hat{X}(\omega) = \sum_{i=0}^{N} x_i\phi_i({\param}(\omega)).
\] The functions $\{\phi_i\}$ are a family of
orthogonal basis in $\mathcal{L}_2(\Omega,\mathcal{F},P)$ satisfying
the relation
\begin{equation}
\Exp{\phi_i\phi_j}:= \int_{\mathcal{D}_{\param}}\hspace{-0.1in}{\basis{i}{\param}\basis{j}{\param} f_{\param}(\param)
\,d\param}  = h_i^2\param_{ij}, \eqnlabel{basisFcn}
\end{equation}
where $\param_{ij}$ is the Kronecker delta, $h_i$ is a constant
term corresponding to $\int_{\mathcal{D}_{\param}}{\phi_i^2f_{\param}(\param)\,d\param}$,
$\mathcal{D}_{\param}$ is the domain of the random variable $\param(\omega)$, and
$f_{\param}(\param)$ is a probability density function for $\param$. Table \ref{table.pc} shows the family of basis functions for random variables with common distributions.
\begin{table}[htbp]
\centering
\begin{tabular}{|c|c|}
\hline
Random Variable $\param$ & $\phi_i(\param)$ of the Wiener-Askey Scheme\\ \hline
Gaussian & Hermite \\
Uniform  & Legendre \\
Gamma   & Laguerre \\
Beta    & Jacobi\\\hline
\end{tabular}
\caption{Correspondence between choice of polynomials and given
distribution of $\param(\omega)$ \cite{xiu:02}.} \label{table.pc}
\end{table}

\subsection{Application to Dynamical Systems with Random Parameters}
With respect to the dynamical system defined in \eqn{contiDyn}, the solution can be approximated by the polynomial chaos expansion as
\begin{align}
	\x(t,\param) \approx \hat{\x}(t,\param) =  \sum_{i=0}^N \x_i(t)\basis{i}{\param},
\end{align}
where the polynomial chaos coefficients $\x_i \in \real^n$. Define $\mo{\Phi}(\param)$ to be
\begin{align}
\mo{\Phi} &\equiv \mo{\Phi}(\param) := \begin{pmatrix}\basis{0}{\param} & \cdots & \basis{N}{\param}\end{pmatrix}^T, \text{ and } \\
\mo{\Phi}_n &\equiv \mo{\Phi}_n(\param) := \mo{\Phi}(\param) \otimes \I{n},
\end{align}
where $\I{n}\in\real^{n\times n}$ is identity matrix. Also define matrix $\X\in\real^{n\times(N+1)}$, with polynomial chaos coefficients $\x_i$, as
\[ \X = \begin{bmatrix} \x_0 & \cdots & \x_N \end{bmatrix}.\]

This lets us define $\hat{\x}(t,\param)$ as 
\begin{align}
\hat{\x}(t,\param) := \X(t)\mo{\Phi}(\param) \eqnlabel{compactX}.
\end{align}
Noting that $\hat{\x} \equiv \vec{\hat{\x}}$, we obtain an alternate form for \eqn{compactX},
\begin{align}
\hat{\x} \equiv  \vec{\hat{\x}} & = \vec{\X\mo{\Phi}}  = \vec{\I{n}\X\mo{\Phi}} = (\mo{\Phi}^T\otimes \I{n})\vec{\X} = \mo{\Phi}_n^T\xpc, \eqnlabel{compactxpc}
\end{align}
where  $\xpc := \vec{\X}$, and $\vec{\cdot}$ is the vectorization operator \cite{horn2012matrix}.

Since $\hat{\x}$ from \eqn{compactxpc} is an approximation, substituting it in \eqn{Acl} we get equation error $\vo{e}$, which is given by 
\begin{align}
\vo{e} &:= \dot{\hat{\x}} - \Acl(\param)\hat{\x}  =  \mo{\Phi}_n^T\xpcdot - \Acl(\param)\mo{\Phi}_n^T\xpc.
\end{align}
Best $\set{L}_2$ approximation is obtained by setting
\begin{align}
\inner{\vo{e}\phi_i} := \Exp{\vo{e}\phi_i}=0, \text{ for } i = 0,1,\cdots,N.
\end{align}
\begin{align}
&\Exp{\Phin{n}\Phin{n}^T}\xpcdot = \Exp{\Phin{n}\Acl\Phin{n}^T}\xpc,\nonumber \\
& \implies \xpcdot  = \Exp{\Phin{n}\Phin{n}^T}^{-1}\Exp{\Phin{n}\Acl\Phin{n}^T}\xpc,  \eqnlabel{pcDynamics}\\
&\text{or } \xpcdot  = \Apc \xpc.
\end{align}
where $\Phin{n}$ and $\Acl$ depend on $\param$ as defined earlier. 

We will need the following result in the rest of the paper.
\begin{proposition} For any vector $\vo{v} \in\real^{N+1}$ and matrix $\vo{M}\in\real^{m\times n}$
\begin{equation}
		\mo{M}(\mo{v}^T\otimes \I{n}) = (\mo{v}^T \otimes \I{m}) (\I{N+1}\otimes\mo{M}), \eqnlabel{prop3}
	\end{equation}
	where $\I{\ast}$ is identity matrix with indicated dimension.
\end{proposition}
\begin{proof}
	\begin{align*}
		\mo{M}(\mo{v}^T \otimes \I{n})  & = (1 \otimes \mo{M})(\mo{v}^T \otimes \I{n})  \\ &= \mo{v}^T \otimes \mo{M} = (\mo{v}^T\I{N+1}) \otimes (\I{m}\mo{M})\\
		&= (\mo{v}^T \otimes \I{m})(\I{N+1}\otimes \mo{M}).	\end{align*}
\end{proof}

\section{Controller Synthesis}
The controller gain $\K(\param)$ can be introduced in the polynomial chaos framework by substituting $\Acl:= \A(\param) + \B(\param)\K(\param)$, in \eqn{pcDynamics}
to get 
\begin{align}
\xpcdot = \Exp{\Phin{n}\Phin{n}^T}^{-1}\left(\Exp{\Phin{n}\A\Phin{n}^T} + \Exp{\Phin{n}\B\K\Phin{n}^T} \right) \xpc \eqnlabel{pcCLP}
\end{align}

Polynomial chaos expansion of $\K(\param)$ can be written as
\begin{align}
\K(\param) &= \sum_{i=0}^{N} \K_i\phi_i(\param), \K_i \in \real^{m\times n}; \nonumber \\
& = \begin{bmatrix}\basis{0}{}\I{m} & \cdots & \basis{N}{}\I{m}\end{bmatrix}
\begin{bmatrix}
\K_0 \\ \vdots \\ \K_N
\end{bmatrix}, \nonumber \\
& = (\mo{\Phi}^T\otimes \I{m}) \VK = \mo{\Phi}^T_m \VK,
\end{align}
where $\VK\in\real^{m(N+1)\times n}$ is the vertical stacking of $\K_i$. The expression $\B\K\Phin{n}^T$ in \eqn{pcCLP} can be simplified using \eqn{prop3} as
\begin{align*}
\B\Phin{m}^T\VK \Phin{n}^T = \B\Phin{m}^T\Phin{m(N+1)}^T \mVK,
\end{align*}
where $\mVK:=\I{N+1} \otimes \VK $. Therefore, 
\begin{align}
\xpcdot = \Exp{\Phin{n}\Phin{n}^T}^{-1}\left(\Exp{\Phin{n}\A\Phin{n}^T} + \Exp{\Phin{n}\B\Phin{m}^T\Phin{m(N+1)}^T}\mVK\right)\xpc. 
\end{align}
%In addition, the stochastic dynamics can be approximated as
%\begin{equation}
%\xdot \approx \left(\A\Phin{n}^T + \B\Phin{m}^T \Phin{m(N+1)}^T\mVK\right)\xpc.
%\end{equation}

Recall that for the dynamical system in \eqn{Acl}, the equilibrium solution is said to possess exponential stability of the $m^{\text{th}}$ mean if $\exists\,\param > 0$ and constants $\alpha>0,\beta>0$ such that $\|\x_0\| < \param$ implies $\forall t\geq t_0$ \cite{bertram1959stability, kats1960stability}
\begin{align}
\Exp{\|\x(t;\x_0,t_0)\|^m_m} \leq \beta \Exp{\|\x_0\|_m^m} e^{-\alpha (t-t_0)}.
\end{align}

It can be shown \cite{raktimCDC2014} that the dynamical system in \eqn{Acl}, with random variables $\param$, is exponentially stable in the $2^{\text{nd}}$ mean, or exponentially stable in the mean square sense (EMS-stable), if $\exists$ a Lyapunov function $V(\x):=\x^T\P\x$, with $\P=\P^T>0$, and $\alpha>0$ such that 
	\begin{align}
		\Exp{\dot{V}} \leq -\alpha \Exp{V}. \eqnlabel{stab_condition}
	\end{align}

\begin{theorem}
The closed-loop system \eqn{Acl} is EMS-stable with controller $\K(\param)$ if $\exists$ $\P=\P^T>0$ and $\alpha > 0$ such that 
\begin{align}
	\set{Y}\AA^T + \AA\set{Y} + \set{W}^T\BB^T + \BB\set{W} + \alpha\set{Y}\Exp{\mo{\Phi}_n\mo{\Phi}_n^T}\leq 0, \eqnlabel{pc_stab}
\end{align}
 where $\set{W} := \I{N+1}\otimes \W$, $\set{Y} := \I{N+1}\otimes \Y$, $\Y := \P^{-1}$, and $\W:=\VK\Y$.
\end{theorem}
\begin{proof}
With $V(\x):=\x^T\P\x$, and $\P = \P^T > 0$, $\dot{V}= \dot{\x}^T\P\x + \x^T\P\dot{\x}$. The term $\dot{\x}^T\P\x$ can be approximated by the polynomial chaos expansion as 
	\begin{align*}
		&\dot{\x}^T\P\x \\
		& \approx \xpc^T\left[\left(\A\Phin{n}^T + \B\Phin{m}^T \Phin{m(N+1)}^T\mVK\right)^T\P\Phin{n}^T\right]\xpc, \\		
		&= \xpc^T\left(\Phin{n}\A^T\P\Phin{n}^T  +  \mVK^T\Phin{m(N+1)}\Phin{m}\B^T\P\Phin{n}^T \right)\xpc,
			\end{align*}
Using \eqn{prop3} we can write $\P\Phin{n}^T = \Phin{n}^T \set{P}$, where  $\set{P}:= \I{N+1}\otimes \P$. Substituting them in $\dot{\x}^T\P\x$ we get
	\begin{align}
	\dot{\x}^T\P\x \approx \xpc^T\left(\Phin{n}\A^T\Phin{n}^T \set{P}  + \mVK^T \Phin{m(N+1)}\Phin{m}\B^T\Phin{n}^T \set{P}\right)\xpc.
	\end{align}
	The Lyapunov function $V:=\x^T\P\x$ can be written as
	\begin{align*}
	V&:=\x^T\P\x = \xpc^T\Phin{n}\P\Phin{n}^T\xpc =  \xpc^T\Phin{n}\Phin{n}^T\set{P}\xpc.
	\end{align*}
	Therefore, $\Exp{\dot{V}}\leq -\alpha\Exp{V}$ is equivalent to
	\begin{align*}
	\Exp{\Phin{n}\A^T\Phin{n}^T}\set{P} + \mVK^T \Exp{\Phin{m(N+1)}\Phin{m}\B^T\Phin{n}^T} \set{P} + (\ast)^T  \leq  -\alpha \Exp{\Phin{n}\Phin{n}^T}\set{P},
	\end{align*}
where $(\ast)^T$ are the symmetric terms.
The above BMI can be convexified using the well known substitutions \cite{bernussou1989linear} $\Y := \P^{-1}$, and $\W:=\VK\Y$. These substitutions can be written in terms of $\set{P}, \mVK, \set{Y}$, and $\set{W}$ as
\begin{align*}
\set{W} & = \I{N+1}\otimes \W = \I{N+1}\I{N+1}\otimes \VK\Y \\
& = (\I{N+1}\otimes \VK)(\I{N+1}\otimes \Y) \\
& = \mVK\set{Y}.
\end{align*}
It is also straightforward to show $\set{P} = \set{Y}^{-1}$ and $\mVK = \set{W}\set{Y}^{-1}$. Substituting these in the above BMI, and pre-post multiplying by $\set{Y}$, we get the result.
\end{proof}

\begin{theorem}
The closed-loop system \eqn{Acl} is EMS-stable with controller $\K(\param)$ if $\exists$ $\P(\param)=\P^T(\param)>0$ and $\alpha > 0$ such that 
\begin{equation}
\set{Y}\vo{M}_1^T + \vo{M}_1\set{Y} + \set{W}^T\vo{M}_2^T + \vo{M}_2\set{W} + \alpha\set{Y}\vo{M}_0 \leq 0,
\end{equation}
where 
\begin{align*}
&\P(\param) :=\Phin{n}^T(\param)\set{P}\Phin{n}(\param),\\
&\set{P} :=\I{N+1}\otimes \P_0, \P_0=\P_0^T>0,\\
& \vo{M}_0 = \Exp{(\Phin{n}\Phin{n}^T)^2},\\
&\vo{M}_1  = \Exp{\Phin{n}\Phin{n}^T\Phin{n}\A\Phin{n}^T}, \\
&\vo{M}_2  = \Exp{\Phin{n}\Phin{n}^T\Phin{n}\B\Phin{m}^T\Phin{m(N+1)}^T}.
\end{align*}
\end{theorem}

\begin{proof}
Define Lyapunov function $V(\x):=\x^T\P(\param)\x$, with $\P(\param) :=\Phin{n}^T(\param)\set{P}\Phin{n}(\param)$, $\set{P} :=\I{N+1} \otimes \P_0$, and $\P_0=\P_0^T>0$. The form for $\P(\param)$ is motivated by the literature on sum-of-square representation of matrix polynomials \cite{scherer2004asymptotically, ichihara2009optimal}, which ensures $\P(\param)>0$. The Lyapunov function can be simplified as
\begin{align*}
V(\x) &= \x^T\P(\param)\x \\
& \approx \xpc^T\Phin{n} \Phin{n}^T (\I{N+1} \otimes \P_0) \Phin{n} \Phin{n}^T \xpc \\
& = \xpc^T\Phin{n} \Phin{n}^T (\I{N+1} \otimes \P_0) \Phin{n} \Phin{n}^T \xpc \\
& = \xpc^T ((\vPhi\vPhi^T)\otimes\I{n})(\I{N+1} \otimes \P_0) ((\vPhi\vPhi^T)\otimes\I{n}) \xpc\\
& =  \xpc^T ((\vPhi\vPhi^T)\otimes\I{n})((\vPhi\vPhi^T)\otimes\P_0)\xpc \\
& = \xpc^T  ((\vPhi\vPhi^T)^2\otimes\I{n})(\I{N+1} \otimes \P_0) \xpc \\
& = \xpc^T (\Phin{n}\Phin{n}^T)^2\set{P} \xpc.
\end{align*}

$\dot{V}= \dot{\x}^T\P\x + \x^T\P\dot{\x}$. The term $\dot{\x}^T\P\x$ can be approximated by the polynomial chaos expansion as

\begin{align*}
 \dot{\x}^T\P(\param)\x\approx \xpc^T\left(\A\Phin{n}^T + \B\Phin{m}^T \Phin{m(N+1)}^T\mVK\right)^T \Phin{n}^T\set{P}\Phin{n} \Phin{n}^T \xpc.
\end{align*}
We next show that $\set{P}\Phin{n} \Phin{n}^T = \Phin{n}\Phin{n}^T\set{P}$.
\begin{align*}
\set{P}\Phin{n} \Phin{n}^T & = (\I{N+1}\otimes\P_0)(\vPhi\otimes \I{n})(\vPhi^T\otimes \I{n}) \\
& = (\I{N+1}\otimes\P_0)((\vPhi\vPhi^T) \otimes \I{n})\\
& = (\vPhi\vPhi^T) \otimes \P_0 \\
& = (\vPhi\vPhi^T)\I{N+1} \otimes \I{n}\P_0\\
& = ((\vPhi\vPhi^T) \otimes \I{n}\I{n}) (\I{N+1} \otimes \P_0) \\
& = \Phin{n}\Phin{n}^T\set{P}.
\end{align*}
Therefore,
\begin{align*}
&\dot{\x}^T\P(\param)\x \\
&\approx \xpc^T\left(\A\Phin{n}^T + \B\Phin{m}^T \Phin{m(N+1)}^T\mVK\right)^T \Phin{n}^T\Phin{n}\Phin{n}^T\set{P}\xpc \\
&= \xpc\left(\Phin{n}\A^T\Phin{n}^T\Phin{n}\Phin{n}^T\set{P} + 
\mVK^T\Phin{m(N+1)}\Phin{m}\B^T\Phin{n}^T\Phin{n}\Phin{n}^T\set{P}\right)\xpc.
\end{align*}
Therefore,  $\Exp{\dot{V}}\leq -\alpha\Exp{V}$ is equivalent to
\[
\vo{M}_1^T\set{P} + \set{P}\vo{M}_1 + \mVK^T\mo{M}_2^T\set{P} + \set{P}\mo{M}_2\mVK + \alpha \vo{M}_0\set{P}\leq 0,
\]
which can be convexified as in Theorem 1 to obtain the result.
\end{proof}

\section{Example}
Here we consider the control of the following nonlinear system
\begin{equation}
\begin{pmatrix}\dot{x}_1 \\ \dot{x}_2 \end{pmatrix} = 
\begin{bmatrix}0 & 1 \\ -1 & (1-x_1^2)\end{bmatrix} \begin{pmatrix}{x}_1 \\ {x}_2 \end{pmatrix} + \begin{bmatrix}0 \\ 1 \end{bmatrix}u. 
\eqnlabel{vdp}
\end{equation}
The above systems is the Van der Pol oscillator with a control input. We transform it to an LPV system by introducing the parameter $\rho:=1-x_1^2$. The objective is to design a state feedback controller $K(\rho)$ that will quadratically stabilize the above system. We restrict stabilization of the set defined by $\x\in[-5,5]^2$. Therefore, $\rho \in [-24,1]$. For the PC LPV algorithm, we assume $\rho \equiv \param \in \uniform{[-24,1]}$, a uniformly distributed random variable over $[-24,1]$. \Fig{plots} shows the state and control trajectories of \eqn{vdp} with three control systems $\K_\text{LTI}, \K_\text{LPV}(\rho)$ and $\K_\text{pcLPV}(\param)$, and were designed with $\alpha = 1$ in the following manner:
\begin{itemize}
\item $\K_\text{LTI}$, from linearized dynamics ($A_\text{LTI},B_\text{LTI}$) about $(0,0)$, satisfying
\begin{align*} 
\Y_\text{LTI}\A_\text{LTI}^T  + \A_\text{LTI} \Y_\text{LTI} + \W^T_\text{LTI}\B_\text{LTI}^T  + \B_\text{LTI} \W_\text{LTI} + \alpha \Y_\text{LTI} \leq 0.
\end{align*}
\item $\K_\text{LPV}(\rho)$, from LPV dynamics 
$$
\A_\text{LPV}(\rho):= \begin{bmatrix}0 & 1 \\ -1 & \rho \end{bmatrix},  \B_\text{LPV}(\rho) := \begin{bmatrix}0 \\ 1 \end{bmatrix}
$$
satisfying
\begin{align*} 
\Y_\text{LPV}(\rho_k)\A_\text{LPV}^T(\rho_k)  + \A_\text{LPV}(\rho_k) \Y_\text{LPV}(\rho_k) + 
\W^T_\text{LPV}(\rho_k)\B_\text{LPV}^T(\rho_k)  + \B_\text{LPV}(\rho_k) \W_\text{LPV}(\rho_k) +  \alpha \Y_\text{LPV}(\rho_k) \leq 0,
\end{align*}
where 
$$
\Y_\text{LPV}(\rho_k) :=\Y_0 + \rho_k \Y_1 > 0, \Y_i = \Y_i^T,
$$
and $\rho_k$ are the samples from $\uniform{[-24,1]}$.
\item $\K_\text{pcLPV}$, from theorem 2, assuming $\rho \equiv \param \in \uniform{[-24,1]}$.
\end{itemize}
 \Fig{plots} shows the state and control trajectories of the nonlinear closed-loop system for initial condition $(5,5)$. $\K_\text{pcLPV}$ is designed with \textit{first order} polynomial chaos expansion and several $\K_\text{LPV}$s are designed with $2, 5, 10$ and $50$ samples from $\uniform{[-24,1]}$.
 
For this problem, we make the following observations.
\begin{enumerate}
\item Increasing the order of the polynomial chaos expansion does not significantly improve controller performance. We are able to achieve high performance with very low order polynomial chaos expansion.
\item As seen from \fig{plots} increasing the number of  samples in the design of $\K_\text{LPV}$, improves the performance, but doesn't quite reach the performance of $\K_\text{pcLPV}$. The computational times for $\K_\text{LPV}$ synthesis are as follows:\\[0.1in] 
\begin{tabular}[h!]{c|c}
\hline Controller & Synthesis Time (s) \\ \hline
$\K_\text{pcLPV}$ (first order PC) & 0.3556 \\
$\K_\text{LPV}$ (2 samples) & 0.4008 \\
$\K_\text{LPV}$ (5 samples) & 0.5463 \\
$\K_\text{LPV}$ (10 samples) & 0.6723 \\
$\K_\text{LPV}$ (50 samples) & 1.9943 \\ \hline
\end{tabular}\\[0.1in]
Thus $\K_\text{pcLPV}$ has a clear advantage over sampled based $\K_\text{LPV}$ design in terms of controller performance and computational complexity.

\item Both $\K_\text{pcLPV}$ and $\K_\text{LPV}$ outperform $\K_\text{LTI}$ as expected.
\end{enumerate}
The controllers were synthesized in \texttt{MATLAB} \cite{matlab} using \texttt{CVX} \cite{grant2008cvx}. 

\section{Summary}
In this paper we presented a new framework to synthesize LPV controllers using polynomial chaos framework. This framework builds on the probabilistic representation of the scheduling variable and the synthesis was done by treating the LPV system as a stochastic linear system. Two synthesis algorithms were presented which correspond to parameter dependent and independent Lyapunov functions. The algorithms were tested on a nonlinear dynamical system and outperformed controllers synthesized using classical LPV design techniques.

\begin{figure}[h]
\includegraphics[width=0.5\textwidth]{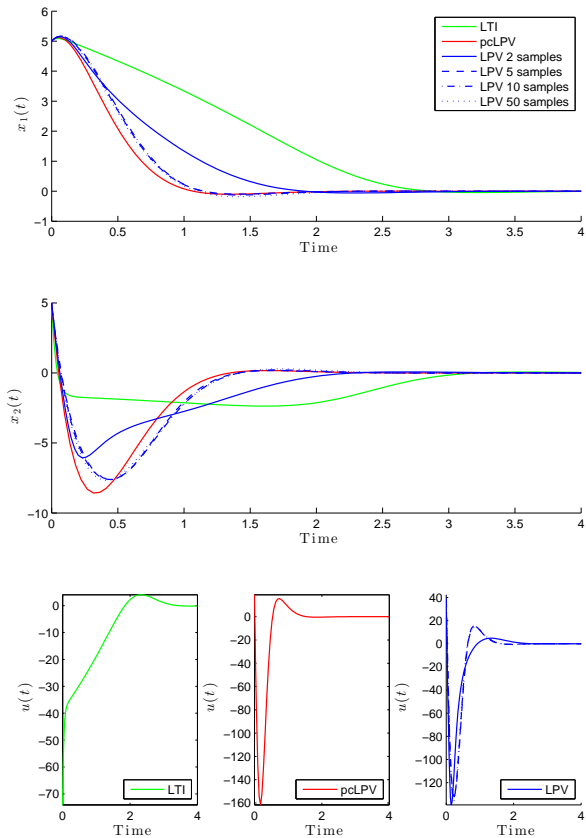}
\caption{State and control trajectories}
\figlabel{plots}
\end{figure}

\bibliographystyle{IEEEtran}
\bibliography{/Users/raktim/Dropbox/myfiles/Latex/raktim}

% Generated by IEEEtran.bst, version: 1.13 (2008/09/30)
\begin{thebibliography}{10}
\providecommand{\url}[1]{#1}
\csname url@samestyle\endcsname
\providecommand{\newblock}{\relax}
\providecommand{\bibinfo}[2]{#2}
\providecommand{\BIBentrySTDinterwordspacing}{\spaceskip=0pt\relax}
\providecommand{\BIBentryALTinterwordstretchfactor}{4}
\providecommand{\BIBentryALTinterwordspacing}{\spaceskip=\fontdimen2\font plus
\BIBentryALTinterwordstretchfactor\fontdimen3\font minus
  \fontdimen4\font\relax}
\providecommand{\BIBforeignlanguage}[2]{{%
\expandafter\ifx\csname l@#1\endcsname\relax
\typeout{** WARNING: IEEEtran.bst: No hyphenation pattern has been}%
\typeout{** loaded for the language `#1'. Using the pattern for}%
\typeout{** the default language instead.}%
\else
\language=\csname l@#1\endcsname
\fi
#2}}
\providecommand{\BIBdecl}{\relax}
\BIBdecl

\bibitem{shamma2012overview}
J.~S. Shamma, ``An overview of lpv systems,'' in \emph{Control of Linear
  Parameter Varying Systems with Applications}.\hskip 1em plus 0.5em minus
  0.4em\relax Springer, 2012, pp. 3--26.

\bibitem{leith2000survey}
D.~J. Leith and W.~E. Leithead, ``Survey of gain-scheduling analysis and
  design,'' \emph{International journal of control}, vol.~73, no.~11, pp.
  1001--1025, 2000.

\bibitem{toker1997complexity}
O.~Toker, ``On the complexity of the robust stability problem for linear
  parameter varying systems,'' \emph{Automatica}, vol.~33, no.~11, pp.
  2015--2017, 1997.

\bibitem{fujisaki2003probabilistic}
Y.~Fujisaki, F.~Dabbene, and R.~Tempo, ``Probabilistic design of lpv control
  systems,'' \emph{Automatica}, vol.~39, no.~8, pp. 1323--1337, 2003.

\bibitem{le2010spectral}
O.~P. Le~Ma{\^\i}tre and O.~M. Knio, \emph{Spectral methods for uncertainty
  quantification: with applications to computational fluid dynamics}.\hskip 1em
  plus 0.5em minus 0.4em\relax Springer, 2010.

\bibitem{wiener}
N.~Wiener, ``The homogeneous chaos,'' \emph{American Journal of Mathematics},
  vol.~60, no.~4, pp. 897--936, Oct. 1938.

\bibitem{volterra}
V.~Volterra, ``{Lecons sur les Equations Integrales et
  Integrodifferentielles},'' \emph{Paris: Gauthier Villars}, 1913.

\bibitem{pcFEM}
R.~G. Ghanem and P.~D. Spanos, \emph{{Stochastic Finite Elements: A Spectral
  Approach}}.\hskip 1em plus 0.5em minus 0.4em\relax New York, NY, USA:
  Springer-Verlag New York, Inc., 1991.

\bibitem{CameronMartin}
R.~H. Cameron and W.~T. Martin, ``{The Orthogonal Development of Non-Linear
  Functionals in Series of {Fourier-Hermit}e Functionals},'' \emph{The Annals
  of Mathematics}, vol.~48, no.~2, pp. 385--392, 1947.

\bibitem{xiu:02}
D.~Xiu and G.~Karniadakis, ``The {Wiener}--askey polynomial chaos for
  stochastic differential equations,'' \emph{SIAM Journal on Scientific
  Computing}, vol.~24, no.~2, pp. 619--644, 2002.

\bibitem{Askey-Polynomials}
R.~Askey and J.~Wilson, ``{Some Basic Hypergeometric Polynomials that
  Generalize Jacobi Polynomials},'' \emph{Memoirs Amer. Math. Soc.}, vol. 319,
  1985.

\bibitem{pcFluids2}
T.~Y. Hou, W.~Luo, B.~Rozovskii, and H.-M. Zhou, ``{Wiener Chaos Expansions and
  Numerical Solutions of Randomly Forced Equations of Fluid Mechanics},''
  \emph{J. Comput. Phys.}, vol. 216, no.~2, pp. 687--706, 2006.

\bibitem{pcFluids4}
D.~Xiu and G.~E. Karniadakis, ``{Modeling Uncertainty in Flow Simulations via
  Generalized Polynomial Chaos},'' \emph{J. Comput. Phys.}, vol. 187, no.~1,
  pp. 137--167, 2003.

\bibitem{pcFluids5}
X.~Wan, D.~Xiu, and G.~E. Karniadakis, ``Stochastic solutions for the
  two-dimensional advection-diffusion equation,'' \emph{SIAM J. Sci. Comput.},
  vol.~26, no.~2, pp. 578--590, 2005.

\bibitem{pcSolids1}
R.~Ghanem and J.~Red-Horse, ``{Propagation of Probabilistic Uncertainty in
  Complex Physical Systems Using a Stochastic Finite Element Approach},''
  \emph{Phys. D}, vol. 133, no. 1-4, pp. 137--144, 1999.

\bibitem{pcSolids2}
R.~Ghanem, ``{Ingredients for a General Purpose Stochastic Finite Elements
  Implementation},'' \emph{Comput. Methods Appl. Mech. Eng.}, vol. 168, no.
  1-4, pp. 19--34, 1999.

\bibitem{hover2006application}
F.~S. Hover and M.~S. Triantafyllou, ``Application of polynomial chaos in
  stability and control,'' \emph{Automatica}, vol.~42, no.~5, pp. 789--795,
  2006.

\bibitem{kim2012generalized}
K.~Kim and R.~D. Braatz, ``Generalized polynomial chaos expansion approaches to
  approximate stochastic receding horizon control with applications to
  probabilistic collision checking and avoidance,'' in \emph{Control
  Applications (CCA), 2012 IEEE International Conference on}.\hskip 1em plus
  0.5em minus 0.4em\relax IEEE, 2012, pp. 350--355.

\bibitem{fisher2009linear}
J.~Fisher and R.~Bhattacharya, ``Linear quadratic regulation of systems with
  stochastic parameter uncertainties,'' \emph{Automatica}, vol.~45, no.~12, pp.
  2831--2841, 2009.

\bibitem{bhattacharya2012linear}
R.~Bhattacharya and J.~Fisher, ``Linear receding horizon control with
  probabilistic system parameters,'' in \emph{Robust Control Design}, vol.~7,
  no.~1, 2012, pp. 627--632.

\bibitem{Dutta2010}
P.~Dutta and R.~Bhattacharya, ``{Nonlinear Estimation with Polynomial Chaos and
  Higher Order Moment Updates},'' in \emph{2010 American Control Conference,
  Marriott Waterfront}, Baltimore, MD, USA, 2010, pp. 3142--3147.

\bibitem{eldred2009comparison}
M.~Eldred and J.~Burkardt, ``Comparison of non-intrusive polynomial chaos and
  stochastic collocation methods for uncertainty quantification,'' \emph{AIAA
  paper}, vol. 976, no. 2009, pp. 1--20, 2009.

\bibitem{horn2012matrix}
R.~A. Horn and C.~R. Johnson, \emph{Topics in Matrix Analysis}.\hskip 1em plus
  0.5em minus 0.4em\relax Cambridge university press, 2012.

\bibitem{bertram1959stability}
J.~Bertram and P.~E. Sarachik, ``Stability of circuits with randomly
  time-varying parameters,'' \emph{Information Theory, IRE Transactions on},
  vol.~5, no.~5, pp. 260--270, 1959.

\bibitem{kats1960stability}
I.~I. Kats and N.~Krasovskii, ``On the stability of systems with random
  parameters,'' \emph{Journal of Applied Mathematics and Mechanics}, vol.~24,
  no.~5, pp. 1225--1246, 1960.

\bibitem{raktimCDC2014}
R.~Bhattacharya, ``Robust state feedback control design with probabilistic
  system parameters,'' in \emph{IEEE CDC}, 2014.

\bibitem{bernussou1989linear}
J.~Bernussou, P.~L.~D. Peres, and J.~C. Geromel, ``A linear programming
  oriented procedure for quadratic stabilization of uncertain systems,''
  \emph{Systems \& Control Letters}, vol.~13, no.~1, pp. 65--72, 1989.

\bibitem{scherer2004asymptotically}
C.~Scherer and C.~Hol, ``Asymptotically exact relaxations for robust {LMI}
  problems based on matrix-valued sum-of-squares,'' in \emph{Proceedings of the
  International Symposium on Mathematical Theory of Networks and Systems
  (MTNS)}, 2004.

\bibitem{ichihara2009optimal}
H.~Ichihara, ``Optimal control for polynomial systems using matrix sum of
  squares relaxations,'' \emph{Automatic Control, IEEE Transactions on},
  vol.~54, no.~5, pp. 1048--1053, 2009.

\bibitem{matlab}
\emph{MATLAB: {H}igh {P}erformance {N}umeric {C}omputation and {V}isualization
  {S}oftware}, The Math Works, 1992.

\bibitem{grant2008cvx}
M.~Grant, S.~Boyd, and Y.~Ye, ``Cvx: Matlab software for disciplined convex
  programming,'' 2008.

\end{thebibliography}

\end{document}